\documentclass[a4paper,11pt]{article}

\input{Preamble.tex}

\begin{document}

\title{Multivariate Analysis of Scheduling Fair Competitions}

\author{
Siddharth Gupta\thanks{Ben-Gurion University of the Negev, Israel. \texttt{siddhart@post.bgu.ac.il}}
\and Meirav Zehavi\thanks{Ben-Gurion University of the Negev, Israel. \texttt{meiravze@bgu.ac.il}}
}

\maketitle

\begin{abstract}
A \emph{fair competition}, based on the concept of envy-freeness, is a non-eliminating competition where each contestant (team or individual player) may not play against all other contestants, but the total difficulty for each contestant is the same: the sum of the initial rankings of the opponents for each contestant is the same. Similar to other non-eliminating competitions like the Round-robin competition or the Swiss-system competition, the winner of the fair competition is the contestant who wins the most games. The \sfair\ (\fair) problem can be used to schedule fair competitions whose infrastructure is known. 
In the \fair\ problem, we are given an infrastructure of a tournament represented by a graph $G$ and the initial rankings of the contestants represented by a multiset of integers $S$. The objective is to decide whether $G$ is \emph{$S$-fair}, i.e., there exists an assignment of the contestants to the vertices of $G$ such that the sum of the rankings of the neighbors of each contestant in $G$ is the same constant $k\in\mathbb{N}$. We initiate a study of the classical and parameterized complexity of \fair\ with respect to several central structural parameters motivated by real world scenarios, thereby presenting a comprehensive picture of it. 
\end{abstract}

\thispagestyle{empty}
\newpage
\pagestyle{plain}
\setcounter{page}{1}
 
\section{Introduction}\label{sec:intro}

Various real life situations require to conduct fair competitions. For illustration, suppose we want to schedule a non-eliminating sports competition in which there are $n$ contestants and $n$ grounds located on the circumference of a circle. As organizers, we want to assign a home ground to each contestant in such a way that every contestant $c$ plays only against $r$ contestants whose home ground is nearest to $c$'s home ground rather than all the contestants. The underlying rationale can be time constraints and also to minimize the travel time for each contestant (similar to the {\sc Traveling Tournament} problem, see e.g.~\cite{DBLP:conf/aaai/HoshinoK11,DBLP:conf/aaai/HoshinoK12}). However, the total difficulty for each contestant should be the same, i.e., the sum of the initial rankings of the opponents for each contestant is the same. We can model this problem as an instance of the \sfair\ (\fair) problem, where we are given an infrastructure of a tournament represented by a graph $G$ and the initial rankings of the contestants represented by a multiset of integers $S$. The objective is to decide whether $G$ is \emph{$S$-fair}, i.e., there exists an assignment of the contestants to the vertices of $G$ such that the sum of the rankings of the neighbors of each contestant in $G$ is the same constant $k\in\mathbb{N}$.
\remove{where we are given a graph $G$ and a multiset of integers $S$. The objective is to decide whether {\em $G$ is $S$-fair}, i.e., there exists a bijective labeling from $V(G)$ to $S$ such that the sum of the labels of the neighbors of every vertex in $G$ is the same constant $k \in \mathbb{N}$. }Here, $k$ is called the {\em $S$-fairness constant}, or simply {\em fairness constant} if $S$ is clear from the context, of $G$. Clearly, the above problem is equivalent to having an $r$-regular graph $G$ with $n$ vertices, one for each ground, and edges connecting each vertex to $r/2$ nearest vertices on the left and $r/2$ nearest vertices on the right, and the objective is to determine whether $G$ is $S$-fair where $S$ is the multiset of the rankings of the contestants  (see Figure~\ref{fi:harareGraph}). As the total difficulty for each contestant in the competition is the same, we refer to such a competition as a \emph{fair competition}.

In general, if we have the infrastructure of the competition (implicitly, like the above example, or explicitly) and we want to schedule a fair competition, we can model the problem as that of determining whether the graph representing the infrastructure of the competition is $S$-fair where $S$ is the multiset of the rankings of the contestants. This situation is very frequently observed in on-line games or in other recurring competitions - in such competitions, the infrastructure of the competition is fixed and the set of contestants keeps changing.

Scheduling competitions and tournaments with different objectives is a well studied problem in the literature. There are, mainly, two fundamental competition designs, with all other designs considered as variations and hybrids. The first one is the elimination (or knockout) competition, in which the contestants are mapped to the leaf nodes of a complete binary tree. Contestants mapped to nodes with same parent compete against each other in a match, and the winner of the match moves up the tree. The contestant who reaches the root node is the winner of the tournament. The second one is the non-eliminating competition, in which no contestant is eliminated after one or few loses, and the winner is decided at the end of all the games by selecting the contestant with largest number of wins.

In recent years, algorithmic perspectives of scheduling both kinds of competitions have received significant attention by the computational social choice community.
We will first discuss elimination competitions, followed by non-eliminating competitions. With respect to elimination competitions, the design of a fair elimination competition under various definitions of being fair has received notable attention~\cite{hwang1982new,DBLP:journals/tamm/Schwenk00,DBLP:conf/atal/VuS10,DBLP:journals/tist/VuS11}. In this context, it is also relevant to mention the {\sc Tournament Fixing} problem. Here, we are given $n$ contestants, an encoding of the outcome of each potential match between every two contestants  as a digraph $D$, and a favorite contestant $v$: the goal is to design an elimination tournament so that $v$ wins the tournament. This problem was introduced by Vu {\em et~al.}~\cite{DBLP:conf/atal/VuAS09}. After this, it was extensively studied~from~both combinatorial and algorithmic (as well as parameterized) points of view~\cite{DBLP:conf/aaai/AzizGMMSW14,DBLP:conf/ijcai/GuptaR0Z18,DBLP:conf/aaai/KimSW16,DBLP:conf/ijcai/KimW15,DBLP:conf/atal/KonickiW19,DBLP:conf/aaai/RamanujanS17,DBLP:conf/wine/StantonW11,DBLP:conf/ijcai/StantonW11,DBLP:conf/aaai/Williams10}. 

With respect to non-eliminating competitions, a round-robin tournament (RRT) is one of the most popular forms, in which each contestant plays every other contestant~\cite{DBLP:journals/eor/ScarfYB09}. A well studied problem regarding RRTs~is~the {\sc Traveling Tournament} problem, where the goal is to design a fair RRT by minimizing the total travel distance for every team~\cite{DBLP:conf/aaai/GoerigkHKW14,DBLP:conf/aaai/HoshinoK11,DBLP:conf/aaai/HoshinoK12,DBLP:journals/scheduling/MeloUR09,DBLP:journals/scheduling/UthusRG12,DBLP:conf/mfcs/XiaoK16}. Another related problem is to design a fair RRT by minimizing the number of ``breaks'' during the tournament~\cite{DBLP:conf/patat/RibeiroU06,DBLP:journals/orl/HofPB10,DBLP:journals/orl/ZengM13}. Despite of being a popular non-eliminating competition, RRTs have some disadvantages. The first disadvantage of this format is the long tournament length, as each contestant plays against all other contestants. From the fairness point of view, a second disadvantage of this format, also mentioned in~\cite{DBLP:journals/eor/ScarfYB09}, is that it favors the {\em strongest} contestants (i.e., the contestants with the highest initial ranking). To see this,~let~$\mathcal{R} = \{r_1, r_2, \ldots, r_n\}$ be the initial rankings of the $n$ contestants, where $r_i$ is the initial ranking of contestant $i$, and let $R$ be the total sum of the rankings. In RRT, the total difficulty faced by contestant $i$ is $R-r_i$, which shows that the total difficulty faced by a contestant increases as we go from the strongest contestants to the weakest contestants. In light of the above disadvantage, motivated by one of the definitions proposed for fairness in~\cite{DBLP:journals/eor/ScarfYB09,DBLP:journals/tamm/Schwenk00,DBLP:journals/tist/VuS11}, and based on a popular fairness concept called \emph{envy-freeness} introduced by Foley~\cite{foley1967resource} in the study of fair division and allocation problems in multi-agent systems (see, e.g.~\cite{DBLP:conf/atal/BarmanG0KN19,DBLP:conf/ijcai/BenabbouCEZ19,DBLP:conf/atal/BeynierBLMRS19,DBLP:conf/sigecom/GhodsiHSSY18}), we define a \emph{fair competition} in an attempt to address both the above disadvantages with RRTs. Here, a \emph{fair competition} is one where each contestant plays with a \emph{subset} of all other contestants, yet the total difficulty for each contestant in the competition is the same. Similar to an envy-free division where no agent feels envy of another agent's share, in a fair competition no contestant feels envy about another contestant's schedule as the total difficulty for each contestant is the same.

Apart from scheduling fair competitions, \fair\ can be used to model other computational problems in social choice. For example, suppose we have $m$ candidates and $n$ jobs, and every job is associated with an integer ``reward''. Every candidate can choose $r$ jobs and every job is chosen by exactly one candidate. Now, we want to get an assignment of the jobs to the candidates such that the total reward collected by every candidate is $k$, for some integer $k$. Then, this is equivalent of having a graph $G$ that is a collection of $m$ stars, each having $r$ leaves, with $n$ total leaves, and the objective is to determine whether $G$ is $S$-fair with the fairness constant $k$ where $S$ is the union of (i) the multiset of rewards and (ii) the multiset $S' =\{k, \ldots, k\}$ containing the element $k$ $m$ times. Intuitively, $S'$ represents the multiset of rewards collected by every candidate. Every star vertex corresponds to a candidate $c$, and its leaves correspond to the jobs $c$ is assigned to. 


\begin{figure}
\centering
\includegraphics[scale=0.5]{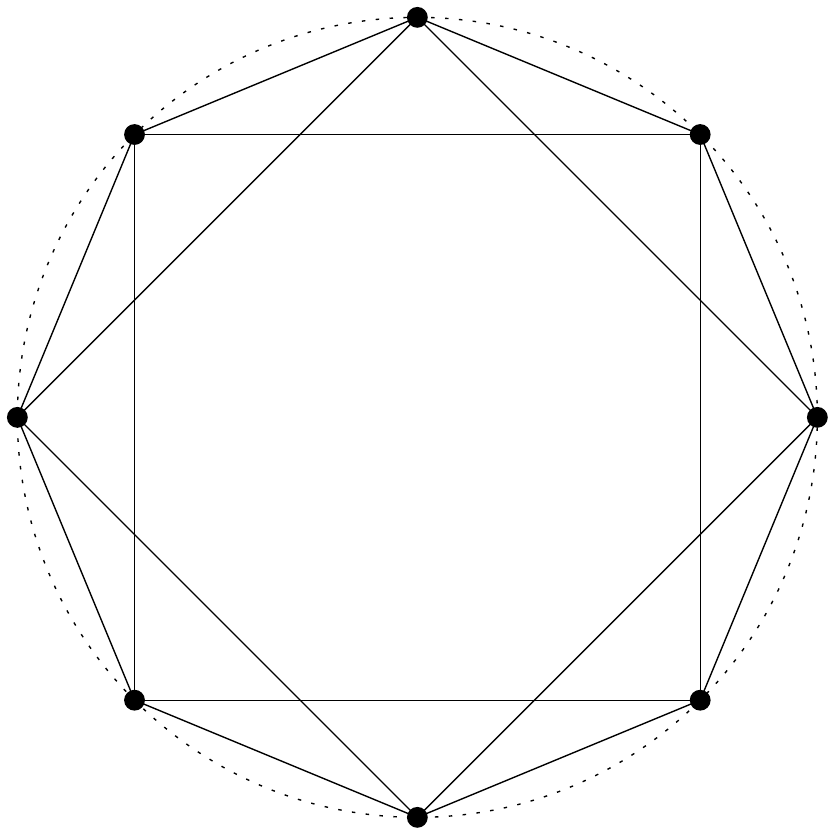}
\caption{Example of a $4$ regular graph where every vertex is connected to $2$ vertices on the left and $2$ on the right.}\label{fi:harareGraph}
\end{figure}

The \fair\ problem can also be used to design semi-magic and magic squares~\cite{wiki:xxx} defined as follows. A {\em semi-magic square} is an $n \times n$ grid $(n \geq 3)$ filled with positive integers from a multiset $I$ such that each cell contains a distinct integer occurrence in $I$ and the sum of integers in each row and each column is the same. A {\em magic square} is a semi-magic square with the additional constraint that the sum of the integers in both the diagonals is also the same and equal to the sum of integers in each row and each column. 

We can model a semi-magic square (and similarly a magic square) as an instance of \fair\ as follows. Let $G$ be a graph with a vertex for each cell in the grid and a vertex for each row and each column, and edges between every cell vertex and its corresponding row and column vertices (see Figure~\ref{fi:semiMagic}). Let the required sum be $k$ and let $S$ be the union of the multiset $I$ and the multiset containing $k-1$ and $1$, each occurring $n$ times. The following observation shows that an $n \times n$ grid can have a semi-magic square filled with integers from the multiset $I$ if and only if $G$ is $S$-fair. 

\begin{observation}
Given an $n \times n$ grid $\mathcal{G}$ $(n \geq 3)$, and a multiset of $n^2$ positive integers $I$, let $k$ be the required sum of any semi-magic square on $\mathcal{G}$, and $G$ and $S$ be the corresponding graph and the multiset respectively. Then $\mathcal{G}$ can have a semi-magic square, filled with integers from the multiset $I$ if and only if $G$ is $S$-fair with fairness constant $k$. 
\end{observation}

\begin{proof}
First, assume that there exists a semi-magic square $M$ on $\mathcal{G}$ filled with integers from the multiset $I$. Then, the sum of integers in each row and each column is $k$. We define an assignment from $V(G)$ to $S$ as follows. Every cell vertex gets the same integer label as the integer it is filled with in $M$. Clearly, the sum of neighbors for every row vertex and every column vertex in $G$ is $k$. Every row is labeled $k-1$ and every column vertex is labeled $1$. Clearly, the sum of neighbors for every cell vertex is $k$ as it is adjacent to exactly one row vertex and exactly one column vertex.

Conversely, let $G$ be $S$-fair. Then, for every vertex in $G$, the sum of the neighbors is $k$. From the construction of $G$, the degree of any row and any column vertex is $n$ and the degree of any cell vertex is $2$. Let $v$ be a vertex having a neighbor whose label is $k-1$. As the sum of the neighbors of $v$ is $k$ and every integer in $S$ is positive, the degree of $v$ must be $2$ and the label of the other neighbor of $v$ is $1$. This implies that $v$ can only be a cell vertex, and all and only the labels $k-1$ and $1$ are used by the $n$ row and $n$ column vertices. So, the labels assigned to cell vertices belong to $I$. As for every row and column vertex, the sum of the neighbors is $k$, by filling every cell in $\mathcal G$ with the label of its corresponding cell vertex, we get a semi-magic square.
\end{proof}

\begin{figure}
\centering
\includegraphics[scale=0.5]{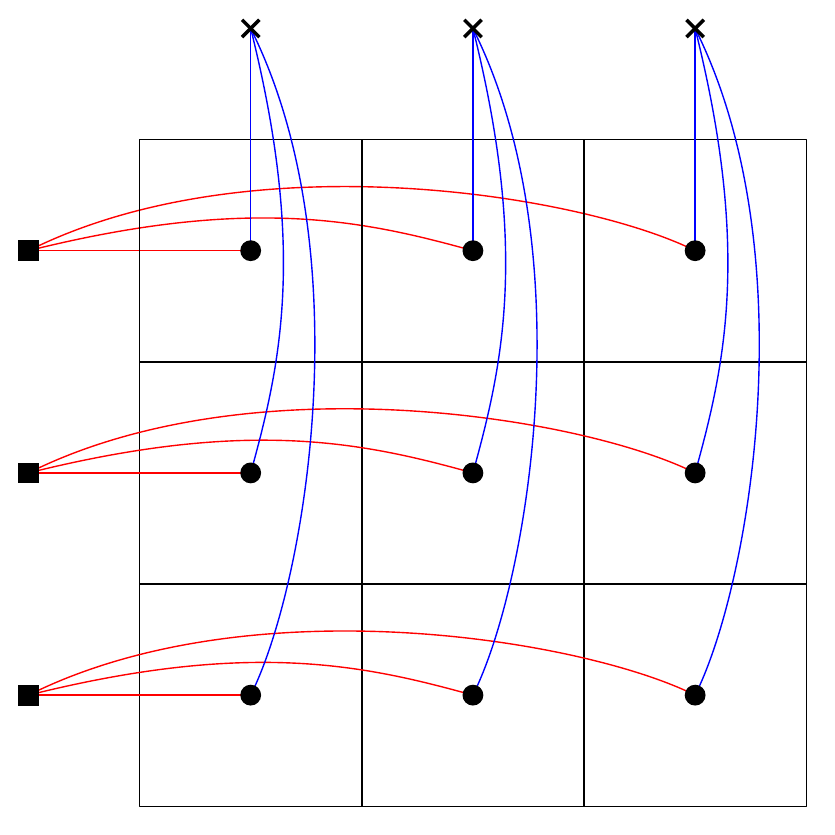}
\caption{Example of the graph $G$ corresponding to a $3 \times 3$ - grid.}\label{fi:semiMagic}
\end{figure}
\subsection{Our Contribution and Methods}

To the best of our knowledge, while \fair\ has been studied extensively from a combinatorial point of view (Section~\ref{subsec:smagic}), close to nothing is known from an algorithmic point of view. We initiate a systematic algorithmic study of \fair. On the one hand, we show \textsf{NP}-hardness results on special graph classes, which imply \textsf{para-NP}-hardness for the problem with respect to several combinations of structural graph parameters. (For basic notions in parameterized complexity, see Section~\ref{sec:prelims}). On the other hand, we show that the problem is \emph{fixed-parameter tractable (\textsf{FPT})} for four different combinations of these parameters. 

The choice of our parameters is motivated by the real world examples from the introduction. In the example of fair competition, we may want every contestant to play only a fraction of the total possible games, which in turn means that the maximum degree $\Delta$ of the infrastructure graph is small compared to the total number of contestants. Similarly, it is likely to happen that a lot of contestants have the same rankings or a lot of jobs have the same rewards, which implies that the number $\alpha$ of unique elements in $S$ is small compared to the total number of contestants or jobs. In the case of job assignment, the underlying graph is a set of stars, which has treewidth $1$ and the size of minimum feedback vertex set is $0$. Moreover, treewidth, feedback vertex set and vertex cover are central parameters in the field of parameterized complexity. 

\begin{table}[!t]
\caption{Summary of our results. Here $\Delta, \tw, \fvs $ and $ \vc$ denote the maximum degree, treewidth, feedback vertex set number and vertex cover number of the input graph, respectively; $\alpha$ denotes the number of distinct elements in $S$. Note that $\tw \leq \fvs \leq \vc$.}
\label{tab:abr} 
\centering
\begin{tabular}{l p{6cm}}\toprule
\textit{Parameters} & \textit{Parameterized Complexity} \\ \midrule
$\tw + \Delta$ & \NPH for $\tw = 3, \Delta = 3$ (also for regular graphs) [Theorem~\ref{the:DegreeTreewidth}]\\
\hline
$\alpha + \Delta$ & \NPH for $\alpha = 3, \Delta = 6$ (also for regular graphs) [Theorem~\ref{the:DeltaAlpha}]\\
\hline
$\fvs + \Delta$ & \NPH for $\fvs = 0, \Delta = 3$ [Theorem~\ref{the:FVSDelta}]\\
\hline
$\fvs + \Delta + \alpha$ & FPT [Theorem~\ref{the:fvsAlphaDelta}]\\
\hline
$\fvs$ & FPT (for regular graphs) [Theorem~\ref{the:fvsrRegular}]\\
\hline
$\vc + \alpha$ & FPT [Theorem~\ref{the:vcAlpha}]\\ \bottomrule
\end{tabular}
\end{table}
Our main results are as follows (summarized in Table~\ref{tab:abr}). First, we show that \fair\ is \NPH for three different graph classes: disjoint unions of $K_{3,3}$'s, disjoint unions of $K_{1,3}$'s and $6$-regular graphs with $3$ distinct labels. Consequently, it is \paraH parameterized by (i) treewidth plus maximum degree, (ii) maximum degree plus feedback vertex set number, and (iii) maximum degree plus the number of distinct labels. The \paraH results hold even for regular graphs when parameterized by either treewidth plus maximum degree or maximum degree plus the number of distinct labels. 

Second, we show that \fair\ is  \FPT parameterized by (i) maximum degree plus feedback vertex set number plus the number of distinct labels, (ii) vertex cover number plus the number of distinct labels, and (iii) feedback vertex set number for regular graphs. We derive some of these results by using insights into \fair\ itself when the input graph is a cycle, a disjoint union of stars, or a connected graph with minimum degree $1$, and Integer Linear Programming.
%

Our choice of parameters also shows several borders of (in-)tractability. For example, the problem is \paraH when parameterized by either $\Delta+\alpha$ or by $\fvs+\Delta$, but becomes \FPT when parameterized by $\fvs+\Delta+\alpha$. Similarly, it is \paraH by $\fvs+\Delta$, but becomes \FPT\ by $\fvs$ for regular graphs. Overall, we give a comprehensive picture of the classical and parameterized complexity of \fair. {\bf For lack of space, some results and proofs marked with an asterisk $(*)$ are omitted or sketched. They are included in the full version, which is available as a supplementary material.}
\subsection{Related Work}\label{subsec:smagic}

The \fair\ problem was first introduced by Vilfred~\cite{vilfred} when $S = \{1, 2, \ldots, n\}$. Such a labeling is called \emph{sigma-labeling} in that paper. The concept of fair scheduling was independently studied by Miller {\em et~al.}~\cite{DBLP:journals/ajc/MillerRS03} in $2003$ under the name \emph{$1$-vertex magic} and by Sugeng {\em et~al.}~\cite{sugeng2009distance} under the name \emph{distance magic labeling}. For recent surveys on distance magic labeling, see~\cite{arumugam2012distance,rupnow2014survey}. The \fair\ problem for a general multiset $S$ was first studied by O'Neal and Slater~\cite{DBLP:journals/siamdm/ONealS13}. In the same paper, they also proved that if a graph $G$ is $S$-fair, then the $S$-fairness constant of $G$ is unique. In~\cite{allLabeling}, Slater proved that \fair\ is {\textsf{NP}-hard}. More recently, Godinho {\em et~al.}~\cite{DBLP:journals/endm/GodinhoSA15} studied the special case of \fair\ where $S$ is a set and not a multiset. They gave a simpler proof for the uniqueness of $S$-fairness constant and also exhibited several families of $S$-fair graphs. Recently, the same set of authors studied a measure called \emph{distance magic index} related to $S$-fair labeling and determined the distance magic index of trees and complete bipartite graphs in~\cite{distanceIndex}. There has also been a long line of studies on other kinds of graph labeling, like {\em \{0,1\}-\fair}, where we consider the closed neighborhood of every vertex instead of open neighborhood (i.e., the vertex itself is also considered in its neighbor set). Another example is {\em vertex-bimagic labeling}, in which there exists two constants $k_1$ and $k_2$ such that the sum of neighbors of every vertex is either $k_1$ or $k_2$. For more information, see the recent survey~\cite{gallian2018dynamic}.
%

\section{Preliminaries}\label{sec:prelims}

\paragraph{Sets and Functions.} Given two multisets $A = \{a_1, a_2, \ldots, a_n\}$ and $B = \{b_1, b_2, \ldots, b_m\}$, their \emph{disjoint union} is the multiset $S = A \uplus B = \{a_1, a_2 , \ldots, a_n, b_1, b_2, \ldots, b_m\}$. For example, let $A = \{1,3,4,5,5\}$ and $B = \{3, 2, 4, 6\}$. Then, $A \uplus B = \{1, 2, 3, 3, 4, 4, 5, 5, 6\}$. Given a multiset $S$, $\alpha(S)$ denotes the number of distinct elements in $S$, and for every $a \in S, \alpha_S(a)$ denotes the number of times $a$ appear in $S$. Given a multiset $A$, $\sum A$ denotes the sum of its elements (in case they are integers), and $|A|$ denotes its size. For any $t \in \mathbb{N}, [t]$ denotes the set $\{1,2,\ldots,t\}$. Given a function $f$ defined on a multiset $A$, $f(A) = \{f(a):a \in A\}$. Let $f: A \rightarrow B$ be a function from a multiset $A$ to a multiset $B$. Then the \emph{restriction of f} to a multiset $A' \subseteq A$ is the function $f|_{A'} : A' \rightarrow B$ given as $f|_{A'}(x) = f(x)$ for every $x \in A'$.

\paragraph{Graphs.} In this paper, we consider only undirected graphs. Given a graph $G$, we denote its vertex set and edge set by $V(G)$ and $E(G)$, respectively. For a vertex $v\in V(G)$, the set of all the neighbors of $v$ in $G$ is denoted by $N_G(v)$, i.e. $N_G(v)=\{u\in V(G)~|~\{u,v\}\in E(G)\}$. The degree of a vertex $v\in V(G)$ in $G$ is denoted by $\degr_G(v)$. When $G$ is clear from the context, we drop the subscript. Given an induced subgraph $H$ of $G$, the set of neighbors of vertices in $H$ which are not in $H$ is denoted by $N_G(H)$, i.e., $N_G(H) = \big(\bigcup_{v \in V(H)} N_G(v)\big) \setminus V(H)$. The maximum and minimum degree of $G$ are denoted by $\Delta(G)$ and $\delta(G)$, respectively. Given a set $V' \subseteq V(G)$, the subgraph of $G$ induced by $V'$ is denoted by $G[V']$. A path on $n$ vertices is denoted by $P_n$. A cycle on $n$ vertices is denoted by $C_n$. A complete bipartite graph with bipartition $A$ and $B$ such that $|A| = m, |B| = n$ is denoted by $K_{m,n}(A, B)$. If $A$ and $B$ are clear from the context, we write $K_{m,n}(A, B)$ as $K_{m,n}$. Given a forest $F$, the set of leaves of $F$ is denoted by $\leaf(F)$. Given a rooted tree $T$, for a vertex $v \in V(T)$, the set of children of $v$ in $T$ is denoted by $\child_T(v)$.

An \emph{$r$-regular graph} is a regular graph where every vertex has degree $r$. An \emph{$n$-star} graph (on $n+1$-vertices) is the complete bipartite graph $K_{1,n}$. Given an $n$-star graph where $n \geq 2$, the \emph{star-vertex} is the unique vertex with degree $n$. The \emph{disjoint union} of two graphs $G_1$ and $G_2$, denoted by $G_1 + G_2$, is the graph with vertex set $V(G_1) \uplus V(G_2)$ and edge set $E(G_1) \uplus E(G_2)$. For any $m \in \mathbb{N}$, we denote the disjoint union of $m$ copies of a graph $G$ by $mG$. Note that, the disjoint union of two or more nonempty graphs is always a disconnected graph. For other standard notations not explicitly defined here, we refer to the book \cite{Diestelbook}.

The {\em treewidth, vertex cover number} and {\em feedback vertex set number} of a graph $G$ are defined as follows.

\begin{definition}[{\bf Treewidth}]
A \emph{tree decomposition} of a graph $G$ is a tree $T$ whose nodes, called \emph{bags}, are labeled by subsets of vertices of $G$. For each vertex $v$, the bags containing $v$ must form a nonempty contiguous subtree of $T$, and for each edge $\{u,v\}$, at least one bag must contain both $u$ and $v$. The \emph{width} of the decomposition is one less than the maximum cardinality of any bag, and the \emph{treewidth} $\tw(G)$ of $G$ is the minimum width of any of its tree decompositions. 
\end{definition}

Based on the definition of treewidth, we have the following observation about disjoint union.

\begin{observation}\label{obs:treewidthDisjoint}
The treewidth of the disjoint union of two vertex-disjoint graphs $G_1$ and $G_2$ is max$\{\tw(G_1), \tw(G_2)\}$.
\end{observation}

\begin{definition}[{\bf Vertex Cover}]
A \emph{vertex cover} of a graph $G$ is a set of vertices in $G$ such that every edge in $G$ has at least one endpoint in the set. We denote the minimum size of a vertex cover of $G$ by $\vc(G)$.
\end{definition}

\begin{definition}[{\bf Feedback Vertex Set}]
A \emph{feedback vertex set} of a graph $G$ is a set of vertices whose removal results in an acyclic graph. We denote the minimum size of a feedback vertex set of $G$ by $\fvs(G)$.
\end{definition}

We will show hardness results from {\sc $3$-Partition} and a variant of SAT called \sat (which were proved to be strongly \NPC and \NPC in \cite{DBLP:books/fm/GareyJ79} and \cite{DBLP:journals/dam/PorschenSSW14}, respectively), defined as follows.

\begin{definition}[{\bf $3$-Partition}]\label{def:3Partition}
Given a multiset $W$ of $n = 3m$ positive integers, for some $m \in \mathbb{N}$, can $W$ be partitioned into $m$ triplets $W_1, W_2,\ldots, W_m$ $($i.e., $W = \biguplus_{i \in [m]}W_i$ and $|W_i| = 3$ for every $i \in [m])$ such that for every $i \in [m], \sum W_i = \sum W/m$?
\end{definition}

\begin{definition}[{\bf $3$-XSAT$^3_+$}]
Given a formula in conjunctive normal form (CNF) where all literals are positive, each clause has size exactly $3$, and each variable occurs exactly $3$ times, does there exist a truth assignment to the variables so that each clause has exactly one true variable?
\end{definition}

The following lemma about a \sat formula will be useful throughout the paper.

\begin{lemma}\label{lem:xsat}
Let $\rho$ be a \sat formula with $n$ variables and $m$ clauses. Then $m = n$. Moreover, $\rho$ is satisfiable only if $n$ is divisible by $3$.
\end{lemma}

\begin{proof}
Let $A$ be the set of $n$ vertices corresponding to $n$ variables and $B$ be the set of $m$ vertices corresponding to $m$ clauses. Consider the bipartite graph $G$ with bipartition $A$ and $B$ and edges between a vertex $x \in A$ and a vertex $c \in B$ if and only if the variable corresponding to $x$ is in clause corresponding to $c$. As every variable appears in exactly $3$ clauses and every clause has exactly $3$ variables, $G$ is $3$-regular. Because $G$ is bipartite, the number of edges in $G$ must be equal to both $3n$ and $3m$ which means that $m = n$. 

Now, assume that $\rho$ is satisfiable, i.e. there exists a truth assignment such that every clause has exactly one true variable. Let $k$ be the number of true variables. Let $v_1$ ans $v_2$ be two vertices in $A$ corresponding to two different true variables. Then $N_G(v_1) \cap N_G(v_2) = \emptyset$, otherwise there is a clause $c$ that has two true variables. Therefore, as $G$ is $3$-regular, total number of clauses satisfied is $3k$. As this number is also $m$ (which equals $n$), this means $k = n/3$. So, $\rho$ is satisfiable only if $n$ is divisible by $3$.
\end{proof}

From the above lemma, it is easy to see that \sat remains \NPC even when $n$ is divisible by $3$. So, in the rest of the paper, we assume that given a \sat formula with $n$ variables and $m$ clauses, $m = n$ and $n$ is divisible by $3$.

\paragraph{Fair Non-Eliminating Tournament.} Given an infrastructure of a tournament represented by a graph $G$ and the initial rankings of the contestants represented by a multiset of integers $S$, $G$ is called \emph{$S$-fair} if there exists an assignment of the contestants to the vertices of $G$ such that the sum of the rankings of the neighbors of each contestant in $G$ is the same. Equivalently, given the infrastructure graph $G$ and the multiset of contestants' rankings $S$ with $|S| = |V(G)|$, $G$ is \emph{$S$-fair} if there exists a bijection $f: V(G) \rightarrow S$ such that for every vertex $v \in V(G), \sum f(N(v)) = k$, where $k$ is a constant called \emph{$S$-fairness constant}. 

For any vertex $v \in V(G), f(v)$ is called the \emph{label} of $v$. We denote the set of all bijective functions that satisfy the above property by $\calM$. In \cite{DBLP:journals/siamdm/ONealS13}, O’Neal and Slater showed that if a graph $G$ is $S$-fair, then its $S$-fairness constant is unique.

Since the infrastructure graph of a tournament is an undirected graph and the initial ranking of several players can be the same, we define the \sfair\ (\fair) problem as follows.

\begin{definition} [{\bf Fair-NET Problem}]\label{def:SMagic}
Given an undirected graph $G$ and a multiset of positive integers $S$ with $|S| = |V(G)|$, is $G$ $S$-fair $($i.e. $\calM \neq \emptyset)$?
\end{definition}

The following observations about $S$-fair graphs follow directly from its definition.

\begin{observation} [{\bf Label Swap}]\label{obs:sameNeighborhood}
Let $G$ be an $S$-fair graph. Let $f \in \calM$. Let $u, v \in V(G)$ such that $N_G(u) = N_G(v)$. Consider $f':V(G) \rightarrow S$, defined as follows. For all $w \in V(G) \setminus \{u,v\}, f'(w) = f(w); f'(u) = f(v); f'(v) = f(u)$. Then $f' \in \calM$.
\end{observation}

\begin{observation} \label{obs:sameLabel}
Let $G$ be an $S$-fair graph. Let $u, v \in V(G)$ such that $N_G(u) = N_G(v)$ and $\{u,v\} \in E(G)$. Then, for all $f \in \calM, f(u) = f(v)$. 
\end{observation}

\begin{observation} \label{obs:rRegularMagicConstant}
Let $G$ be an $r$-regular $S$-fair graph. Then the $S$-fairness constant is equal to $r \cdot \sum S/ |V(G)|$.
\end{observation}

\begin{observation}\label{obs:disjointMagic}
Given two graphs $G_1$ and $G_2$ and a multiset of positive integers $S$, $G_1 + G_2$ is $S$-fair if and only if $G_1$ is $S_1$-fair and $G_2$ is $S_2$-fair with the same fairness constant for some $S_1, S_2 \subseteq S$ such that $S_1 \uplus S_2 = S$. 
\end{observation}

\begin{observation}\label{obs:bipartiteMagic}
Given a complete bipartite graph $K_{m,n}(A, B)$ and a multiset of positive integers $S$, $K_{m,n}$ is $S$-fair if and only if there exists a bijection $f: V(K_{m,n}) \rightarrow S$ such that $\sum f(A) = \sum f(B) = \sum S /2$.
\end{observation}

\paragraph{Integer Linear Programming.} In the {\sc Integer Linear Programming Feasibility} (ILP) problem, the input consists of $p$ variables $x_1, x_2, \ldots, x_p$ and a set of $m$ inequalities of the following form:
	\[\begin{array}{*{9}{@{}c@{}}}
		a_{1,1}x_1 & + & a_{1,2}x_1 &+ & \cdots & + & a_{1,p}x_p & \leq & y_1 \\
		a_{2,1}x_1 & + & a_{2,2}x_2 & + & \cdots & + & a_{2,p}x_p & \leq & y_2 \\
		\vdots    &   & \vdots    &   &        &   & \vdots    &   & \vdots \\
		a_{m,1}x_1 & + & a_{m,2}x_2 & + & \cdots & + & a_{m,p}x_p & \leq & y_m \\
	\end{array}\]
where every coefficient $a_{i_j}$ and $y_i$ is required to be an integer. The task is to check whether there exists an assignment of integer values for every variable $x_i$ so that all inequalities are satisfiable. The following theorem about the tractability of the ILP problem will be useful throughout Section~\ref{sec:fpt}.

\begin{theorem}[\cite{DBLP:journals/mor/Kannan87,DBLP:journals/mor/Lenstra83,DBLP:journals/combinatorica/FrankT87}]\label{the:runningTimeILP}
The ILP problem with $p$ variables is \FPT parameterized by $p$. 
\end{theorem}

\paragraph{Parameterized Complexity.} A problem $\Pi$ is a {\em parameterized} problem if each problem instance of $\Pi$ is associated with a {\em parameter} $k$. For simplicity, we denote a problem instance of a parameterized problem $\Pi$ as a pair $(I,k)$ where the second argument is the parameter $k$ associated with $I$. The main objective of the framework of Parameterized Complexity is to confine the combinatorial explosion in the running time of an algorithm for an \textsf{NP}-hard parameterized problem $\Pi$ to depend only on $k$. In particular, a parameterized problem $\Pi$ is {\em fixed-parameter tractable} (\textsf{FPT}) if any instance $(I, k)$ of $\Pi$ is solvable in time $f(k)\cdot |I|^{\OO(1)}$, where $f$ is an arbitrary computable function of $k$. Moreover, a parameterized problem $\Pi$ is {\em \paraH} if it is \textsf{NP}-hard for some fixed constant value of the parameter $k$. For more information on Parameterized Complexity, we refer the reader to books such as \cite{DBLP:series/txcs/DowneyF13,DBLP:books/sp/CyganFKLMPPS15}.
\section{\textsf{Para-NP}-hardness Results}\label{sec:hardness}

In this section, we exhibit the \textsf{para-NP}-hardness of the \fair\ problem with respect to several structural graph parameters. We start with a \textsf{para-NP}-hardness result with respect to the parameter $\tw + \Delta$.


\begin{theorem}\label{the:DegreeTreewidth}
The \fair\ problem is \NPH for $3$-regular graphs with treewidth $3$. In particular, it is \paraH parameterized by $\tw + \Delta$, even for regular graphs.
\end{theorem}

\begin{proof}
We present a reduction from {\sc $3$-Partition}. Given a multiset $W$ of $n = 3m$ positive integers, for some $m \in \mathbb{N}$, we create two instances of \fair\ based on the value of $m$ as follows (see Figure~\ref{fi:theorem1}). \smallskip\\
{\bf Case~1 [When $m$ is a multiple of $2$]:} In this case, we create an instance $(G, S)$ of \fair\ where $G = (m/2)K_{3,3}$ and $S = W$. Note that $G$ is a $3$-regular graph. Since $\tw(K_{3,3}) = 3$, by Observation~\ref{obs:treewidthDisjoint}, $\tw(G) = 3$. Let $V(G) = \biguplus_{i \in [m/2]}V_i$ where $V_i = A_i \cup B_i$ is the vertex set of the $i$-th copy of $K_{3,3}$ with bipartition $A_i$ and $B_i$. We now prove that $W$ is a \yes instance of {\sc $3$-Partition} if and only if $G$ is $S$-fair.

Assume first that $W$ is a \yes instance of {\sc $3$-Partition}. Let $W_1, W_2, \ldots, W_m$ be a corresponding partition of $W$. Then, by Definition~\ref{def:3Partition}, for every $i \in [m], \sum W_i = \sum W/m$. Let $f: V(G) \rightarrow S$ be a bijective function defined as follows. For every $i \in [m/2],$ let $f(A_i) = W_i$ and $f(B_i) = W_{m/2+i}$ (the internal labeling within $A_i$ and $B_i$ is arbitrary). So, for every $i \in [m/2], \sum f(A_i) = \sum f(B_i) = \sum W/m$. Thus, by Observations~\ref{obs:disjointMagic} and~\ref{obs:bipartiteMagic}, $G = (m/2)K_{3,3}$ is $S$-fair.

Conversely, let $G = (m/2)K_{3,3}$ be $S$-fair. Then, by Observations~\ref{obs:disjointMagic} and~\ref{obs:bipartiteMagic}, there exists a bijection $f:V(G) \rightarrow S$ such that for every $i \in [m/2], \sum f(A_i) = \sum f(B_i) = \sum S/m = \sum W/m$. Thus, $\{f(A_1), f(A_2), \ldots, f(A_{m/2}), f(B_1), \ldots, f(B_{m/2})\}$ is a partition of $W$ satisfying the required property, so $W$ is a \yes instance of {\sc $3$-Partition}. \smallskip\\
{\bf Case~2 [When $m$ is not a multiple of $2$]:} Without loss of generality, we can assume that every element in $W$ is greater than $1$ as otherwise we can get an equivalent instance of {\sc $3$-Partition} by adding 1 to all the elements of $W$. Let $sum = \sum W/m$ be the required sum of every subset. As every element in $W$ is greater than $1$, $sum \geq 6$. In this case, we create an instance $(G, S)$ of \fair\ where $G = \big((m+1)/2\big)K_{3,3}$ and $S = W \uplus \{sum-2, 1, 1\}$. Note that $G$ is a $3$-regular graph and $\tw(G) = 3$. Let $V(G) = \biguplus_{i \in [m+1/2]}V_i$ where $V_i = A_i \cup B_i$ is the vertex set of the $i$-th copy of $K_{3,3}$ with bipartition $A_i$ and $B_i$. We now prove that $W$ is a \yes instance of {\sc $3$-Partition} if and only if $G$ is $S$-fair.

Assume first that $W$ is a \yes instance of {\sc $3$-Partition}. Let $W_1, W_2, \ldots, W_m$ be the corresponding partition of $W$. Then, by Definition~\ref{def:3Partition}, for every $i \in [m], \sum W_i = \sum W/m = sum$. Let $W_{m+1} = \{sum-2, 1, 1\}$. Clearly, $\sum W_{m+1} = sum$. As $S = W \biguplus \{sum-2, 1, 1\}$, $S = \biguplus_{i \in [m+1]}W_i$. Let $f: V(G) \rightarrow S$ be a bijective function defined as follows. For every $i \in [(m+1)/2],$ let $f(A_i) = W_i$ and $f(B_i) = W_{(m+1)/2+i}$ (the internal labeling within $A_i$ and $B_i$ is arbitrary). So, for every $i \in [(m+1)/2], \sum f(A_i) = \sum f(B_i) = sum$. Thus, by Observations~\ref{obs:disjointMagic} and~\ref{obs:bipartiteMagic}, $G = \big((m+1)/2\big)K_{3,3}$ is $S$-fair.

Conversely, let $G = (m+1)/2K_{3,3}$ be $S$-fair. Then, by Observations~\ref{obs:disjointMagic} and~\ref{obs:bipartiteMagic}, there exists a bijection $f:V(G) \rightarrow S$ such that for every $i \in [(m+1)/2], \sum f(A_i) = \sum f(B_i) = \sum S/(m+1) = \sum W/m$. Without loss of generality, let $A_s$ be the set containing $sum-2$, for some $s \in [(m+1)/2]$. As $\sum A_s = sum, |A_s| = 3$ and all the elements in $W$ are greater than $1$, necessarily $A_s = \{sum-2, 1, 1\}$. As $S = W \biguplus \{sum-2, 1, 1\}$, we get that $\{f(A_1), \ldots, f(A_{s-1}), f(A_{s+1}), \ldots, f(A_{(m+1)/2}), f(B_1),$ $\ldots, f(B_{(m+1)/2})\}$ is a partition of $W$ satisfying the required property, so $W$ is a \yes instance of {\sc $3$-Partition}.
\end{proof}

\begin{figure}
\centering
\includegraphics[page=1,scale=0.5]{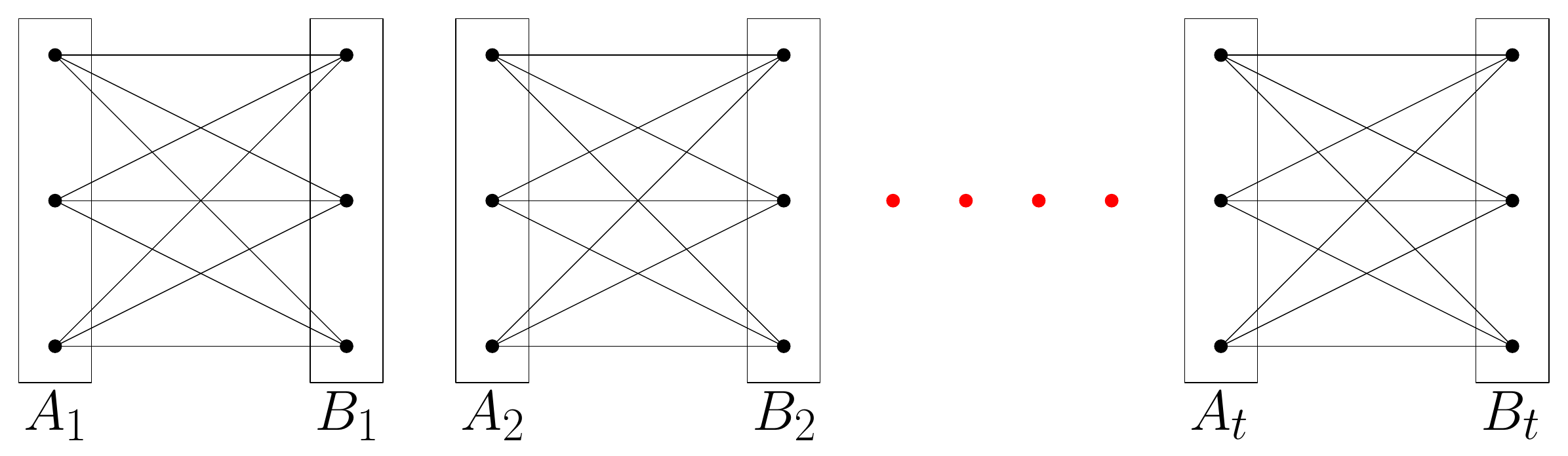}
\caption{Example of the graph $G$ built in the reduction of Theorem~\ref{the:DegreeTreewidth}. $t = m/2$ for Case~1 and $t=(m+1)/2$ for Case~2.}\label{fi:theorem1}
\end{figure}

We now proceed with the \textsf{para-NP}-hardness result with parameter $\fvs + \Delta$.

\begin{theorem}\label{the:FVSDelta}
The \fair\ problem is \NPH for forests with $\Delta = 3$. Since forests have $\fvs = 0$, \fair\ is \paraH parameterized by $\fvs + \Delta$.
\end{theorem}

\begin{proof}
We present a simple reduction from {\sc $3$-Partition}. Given a multiset $W$ of $n = 3m$ positive integers, for some $m \in \mathbb{N}$, let $sum = \sum W/m$ be the required sum of every subset. We create an instances $(G, S)$ of \fair\ where $G = mK_{1,3}$ and $S = W \biguplus \{s_1=sum, s_2=sum, \ldots, s_m=sum\}$.  Note that $G$ is a forest with $\Delta(G) = 3$. Let $V(G) = \biguplus_{i \in [m]}V_i$ where $V_i = \{v_i\} \cup B_i$ is the vertex set of the $i$-th copy of $K_{1,3}$ with $B_i$ being the set of leaves. See Figure~\ref{fi:theorem2}. We now prove that $W$ is a \yes instance of {\sc $3$-Partition} if and only if $G$ is $S$-fair.

Assume first that $W$ is a \yes instance of {\sc $3$-Partition}. Let $W_1, W_2, \ldots, W_m$ be the corresponding partition of $W$. Then, by Definition~\ref{def:3Partition}, for every $i \in [m], \sum W_i = sum$. Let $f: V(G) \rightarrow S$ be a bijective function defined as follows. For every $i \in [m]$, let $f(B_i) = W_i$ and $f(v_i) = sum$ (the internal labeling of $B_i$ is arbitrary). So, for every $i \in [m], f(v_i) = \sum f(B_i) = sum$. Thus, by Observations~\ref{obs:disjointMagic} and~\ref{obs:bipartiteMagic}, $G = mK_{1,3}$ is $S$-fair.

Conversely, let $G = mK_{1,3}$ be $S$-fair. Then, by Observations~\ref{obs:disjointMagic} and~\ref{obs:bipartiteMagic}, there exists a bijection $f:V(G) \rightarrow S$ such that for every $i \in [m], f(v_i) = \sum f(B_i) = \sum S/(m+1) = sum$. As $S = W \biguplus \{s_1=sum, s_2=sum, \ldots, s_m=sum\}$, we get that $\{f(B_1),$ $\ldots, f(B_m)\}$ is a partition of $W$ satisfying the required property, so $W$ is a \yes instance of {\sc $3$-Partition}.
\end{proof}

\begin{figure}
\centering
\includegraphics[page=2,scale=0.5]{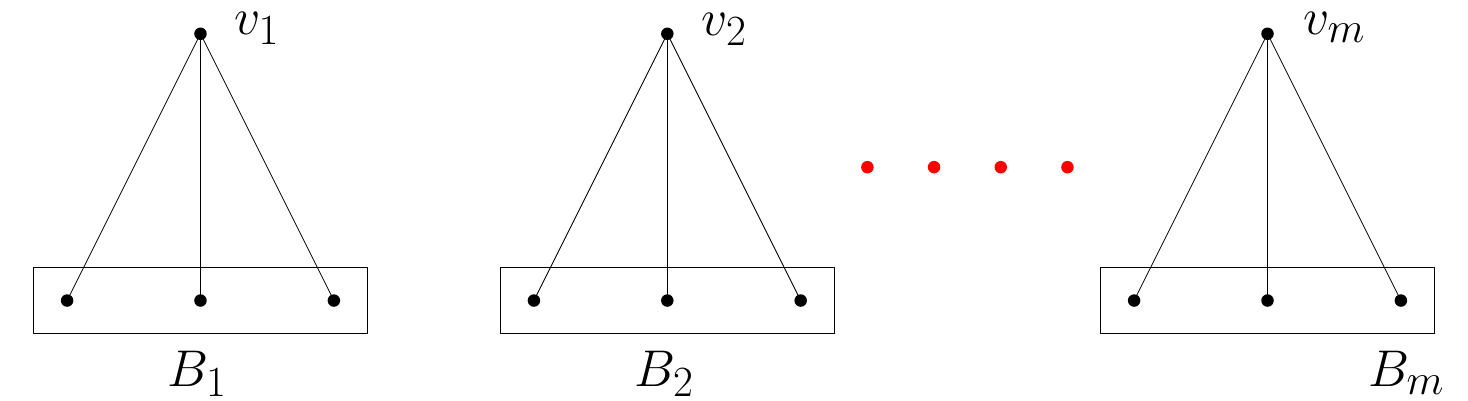}
\caption{Example of the graph $G$ built in the reduction of Theorem~\ref{the:FVSDelta}.}\label{fi:theorem2}
\end{figure}

Finally, we give the \textsf{para-NP}-hardness result with parameter $\alpha + \Delta$. (Recall that $\alpha$ is the number of distinct integers in the input multiset.)

\begin{theorem}\label{the:DeltaAlpha}
The \fair\ problem in \NPH for $6$-regular graphs with $3$ distinct labels. In particular, it is \paraH parameterized by $\alpha + \Delta$, even for regular graphs.
\end{theorem}

\begin{proof}
We present a reduction from \sat. Given a \sat formula $\rho$ with $n$ variables and $n$ clauses, we create an instance $(G,S)$ of \fair\ as follows. Suppose that the variables are indexed by $1, 2, \ldots, n$ and so do the clauses. For every $i \in [n]$, the variable gadget in $G$ consists of a single vertex $x_i$ called a variable vertex. Let $A$ be the set of all the variable vertices. For every $i \in [n]$, the clause gadget in $G$ consists of $15$ vertices $c^1_i, c^2_i, \ldots, c^{15}_i$. For every $i \in [n]$, we add the following edges between these $15$ vertices in the clause gadget in $G$ (see Figure~\ref{fi:theorem3}):
\begin{itemize}
	\item $\forall j \in [3], \{c^1_i, c^{1+j}_i\}$. [Edge set of $K_{1,3}(\{c^1_i\}, \{c^2_i, c^3_i, c^4_i\})$].
	\item $\forall j \in [3], \{c^8_i, c^{8+j}_i\}$. [Edge set of $K_{1,3}(\{c^8_i\}, \{c^9_i, c^{10}_i, c^{11}_i\})$].
	\item $\{c^2_i, c^3_i\}, \{c^3_i, c^4_i\}, \{c^4_i, c^2_i\}$. [Edge set of complete graph on $\{c^2_i, c^3_i, c^4_i\}$].
	\item $\{c^9_i, c^{10}_i\}, \{c^{10}_i, c^{11}_i\}, \{c^{11}_i, c^9_i\}$. [Edge set of complete graph on $\{c^9_i, c^{10}_i, c^{11}_i\}$].
	\item $\forall j \in \{2,3,4\}, \forall k \in \{5,6,7\}, \{c^j_i, c^k_i\}$. [Edge set of $K_{3,3}(\{c^2_i, c^3_i, c^4_i\}, \{c^5_i, c^6_i, c^7_i\})$].
	\item $\forall j \in \{9,10,11\}, \forall k \in \{12,13,14\}, \{c^j_i, c^k_i\}$. [Edge set of $K_{3,3}(\{c^9_i, c^{10}_i, c^{11}_i\}, \{c^{12}_i, c^{13}_i, c^{14}_i\})$].
	\item $\forall j \in \{5,6,7, 12, 13, 14\}, \{c^{15}_i, c^j_i\}$. [Edge set of $K_{1,6}(\{c^{15}_i\}, \{c^5_i, c^6_i, c^7_i, c^{12}_i, c^{13}_i, c^{14}_i\})$].
\end{itemize}
We now explain how we connect the variable and the clause gadgets. For every variable vertex $x_i$, let $j, k$ and $l$ be the indices of the clauses where the $i$-th variable appears. Then, we add the $6$ edges $\{x_i, c^1_j\}, \{x_i, c^8_j\}, \{x_i, c^1_k\}, \{x_i, c^8_k\}, \{x_i, c^1_l\}$ and $\{x_i, c^8_l\}$ to $G$. This completes the construction of $G$. Note that $|V(G)| = n + 15n = 16n$. We now define $S$ as the multiset containing $3$ distinct labels $1, 2$ and $4$ with $\alpha_S(1) = 2n/3, \alpha_S(2) = 15n$ and $\alpha_S(4) = n/3$. From the above construction, it is easy to see that $G$ is a $6$-regular graph and $S$ contains $3$ distinct labels. By Observation~\ref{obs:rRegularMagicConstant}, the $S$-fairness constant $k = 12$. 

In what follows, we will set a variable to true if and only if the label of the corresponding variable vertex is $4$ and false otherwise. We now prove that $\rho$ is a satisfiable if and only if $G$ is $S$-fair.

Assume first that $\rho$ is satisfiable. Let $A'$ be the subset of variable vertices for which the corresponding variables are true. From Lemma~\ref{lem:xsat}, $|A'| = n/3$. Let $f: V(G) \rightarrow S$ be a bijective function defined as follows:
\begin{inparaenum}[(i)]
	\item for all $v \in A',  f(v) = 4$;
	\item for all $v \in A\setminus A', f(v) = 1$;
	\item for all $v \in V(G)\setminus A, f(v) = 2$.
\end{inparaenum} 
For every $i \in [n]$, let $B$ be the set containing vertices the $c^1_i$ and $c^8_i$. From the construction of $G$, only vertices from $B$ have neighbors in $A$, so for all $v \in V(G)\setminus B, f(N_G(v)) = \{2,2,2,2,2,2\}$. As every clause has exactly one true variable and two false variables, for every vertex $v \in B, f(N_G(v)) = \{4, 1, 1, 2, 2, 2\}$. So, for all $v \in V(G), \sum f(N_G(v)) = 12$. Hence, $G$ is $S$-fair.

Conversely, let $G$ be $S$-fair. Then, there exists a bijection $f:V(G) \rightarrow S$ such that for all $v \in V(G), \sum f(N_G(v)) = 12$. Note that, for every $i \in [n], N_G(c^5_i) = N_G(c^6_i) = N_G(c^7_i)$ and $\{c^5_i, c^6_i\}, \{c^6_i, c^7_i\}, \{c^7_i, c^5_i\} \in E(G)$, so by Observation~\ref{obs:sameLabel}, $f(c^5_i) = f(c^6_i) = f(c^7_i)$. Similarly, $f(c^{12}_i) = f(c^{13}_i) = f(c^{14}_i)$, $f(c^2_i) = f(c^3_i) = f(c^4_i)$ and $f(c^9_i) = f(c^{10}_i) = f(c^{11}_i)$. 

Consider the vertex $c^{15}_i$, for any $i \in [n]$. As $\sum f(N_G(c^{15}_i)) = 12$ and $N_G(c^{15}_i) = \{c^5_i, c^6_i, c^7_i, c^{12}_i,$ $c^{13}_i, c^{14}_i\}$, we have that $3f(c^5_i) + 3f(c^{12}_i) = 12$ (by the equalities above). Therefore, $f(c^5_i) + f(c^{12}_i) = 4$ which necessarily implies that $f(c^5_i) = f(c^{12}) = 2$. Now, consider the vertex $c^5_i$. As $\sum f(N_G(c^5_i)) = 12$ and $N_G(c^5_i) = \{c^6_i, c^7_i, c^2_i, c^3_i, c^4_i, c^{15}_i\}$, we have that $3f(c^2_i) + f(c^{15}_i) = 8$ (by equalities above). Therefore, necessarily $f(c^2_i) = f(c^{15}_i) = 2$. Similarly, $f(c^9_i) = 2$. Now, consider the vertex $c^2_i$. As $\sum f(N_G(c^2_i)) = 12$ and $N_G(c^2_i) = \{c^1_i, c^3_i, c^4_i, c^5_i, c^6_i, c^7_i\}$, necessarily $f(c^1_i) = 2$. Similarly, $f(c^8_i) = 2$. 

So far, we conclude that all and only the occurrences of integer $2$ in $S$ are used to label the vertices of the clause gadgets. Finally, consider the vertex $c^1_i$. Let $c^1_i$ is adjacent to variable vertices $x_j, x_k$ and $x_l$. As $\sum f(N_G(c^1_i)) = 12$ and $N_G(c^1_i) = \{c^2_i, c^3_i, c^4_i, x_j, x_k, x_l\}$, we get that $f(x_j) + f(x_k) + f(x_l) = 6$. The only solution to this equation is $\{4,1,1\}$ as the remaining labels are $4$ and $1$. Without loss of generality, let $f(x_j) = 4, f(x_k) = f(x_l) = 1$. Recall that, we set a variable to true if and only if the label of the corresponding variable vertex is $4$ and false otherwise, so we assign variable corresponding to $x_j$ as true and variables corresponding to $x_k$ and $x_l$ as false. It is easy to see that a clause has exactly one true variable. As the number of times $4$ appear in $S$ is $n/3$ and every variable appears exactly in $3$ clauses, the number of satisfied clauses is $n$. Hence $\rho$ is satisfiable. 
\end{proof}

\begin{figure}
\centering
\includegraphics[page=3,scale=0.5]{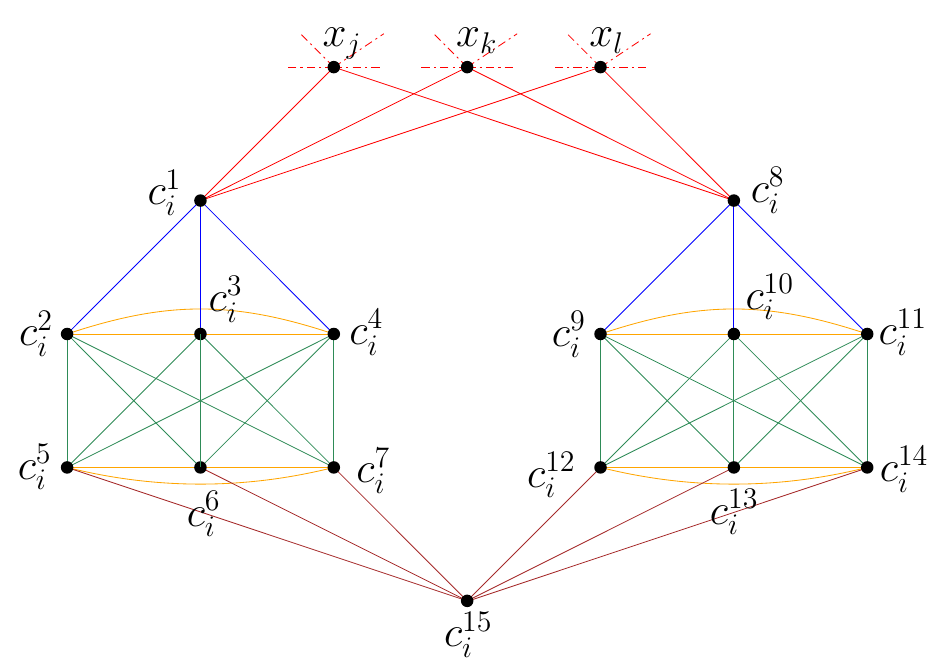}
\caption{Example of the clause gadget in $G$ for a clause $c_i = x_j \vee x_k \vee x_l$ built in the reduction of Theorem~\ref{the:DeltaAlpha}.}\label{fi:theorem3}
\end{figure}
\section{FPT Algorithms}\label{sec:fpt}

In this section, we develop \FPT algorithms for the \fair\ problem with respect to several structural graph parameters. We begin by giving a observation for a graph $G$ to be $S$-fair when $G$ contains isolated vertices.

\begin{observation}
Let $G$ be a graph with $\delta(G) = 0$ and $S$ be a multiset of positive integers. Then, $G$ is $S$-fair if and only if for every vertex $v \in V(G), \deg_G(v) = 0$ (i.e., $G$ contains only isolated vertices).
\end{observation}

Due to the above observation, in the rest of this section, we assume that $G$ does not contain any isolated vertices. We now give conditions that a graph $G$ must satisfy to be $S$-fair when $G$ is a cycle or $\delta(G) = 1$.

\begin{lemma}\label{lem:degOneMagic}
Let $G$ be a connected graph with $\delta(G) = 1$ and $S$ be a multiset of positive integers. Then, $G$ is $S$-fair only if $G$ is a star.
\end{lemma}

\begin{proof}
Assume that $G$ is $S$-fair. Let $f$ be a function in $\calM$ and $k$ be the $S$-fairness constant. Let $v \in V(G)$ be a vertex of degree $1$, and let $u$ be the neighbor of $v$. As $\sum f(N_G(v)) = k$, we get that $f(u) = k$. Assume for contradiction that $G$ is not a star. Then, $|V(G)| \geq 4$, otherwise $G$ is a star. Let $w$ be a neighbor of $u$ other than $v$ such that $\degr_G(w) \geq 2$. (If no such $w$ exists, then $G$ is a star.) As $\sum f(N_G(w)) = k$ and $f(u) = k$, we get that $\sum f(N_G(w) \setminus \{u\}) = 0$. Since $N_G(w) \setminus \{u\} \neq \emptyset$ and labels are positive, this is a contradiction.
\end{proof}

\begin{lemma}\label{lem:cycleDistanceMagic}
Let $G$ be a cycle graph on $n$ vertices and $S$ be a multiset of positive integers. Let $k$ be the required $S$-fairness constant. Then:

\begin{itemize}
	\item If $n \bmod 4 = 0$, then $G$ is $S$-fair if and only if $S$ contains $4$ labels $a, b, k-a, k-b$ with $\alpha_S(a) = \alpha_S(b) = \alpha_S(k-a) = \alpha_S(k-b) = n/4$, for some $a,b \in \mathbb{N}$ such that $a,b < k$.
	\item If $n \bmod 4 \neq 0$, then $G$ is $S$-fair if and only if $S$ contains only one label, $k/2$, with $\alpha_S(k/2) = n$.
\end{itemize}
\end{lemma}

\begin{proof}
Assume first that $G$ is $S$-fair. Denote $V(G) = \{v_1, v_2, \ldots, v_n\}$ in the cyclic order, and let $f \in \calM$. Then, for every $v \in V(G), \sum f(N(v)) = k$. As $G$ is a cycle, for every $i \in [n], N(v_i) = \{v_{(i-1)\bmod n}, v_{(i+1)\bmod n}\}$. Thus, we have that for every $i \in [n], f(v_{(i-1)\bmod n}) + f(v_{(i+1)\bmod n}) = k$. If we expand these equations, we get the following:

\begin{equation}
f(v_1) + f(v_3) = f(v_3) + f(v_5), \ldots, f(v_{n-3}) + f(v_{n-1}) = f(v_{n-1}) + f(v_1).
\end{equation}
\begin{equation}
f(v_2) + f(v_4) = f(v_4) + f(v_6), \ldots, f(v_{n-2}) + f(v_n) = f(v_n) + f(v_2).
\end{equation}
\begin{equation}
f(v_3) + f(v_5) = f(v_5) + f(v_7), \ldots, f(v_{n-1}) + f(v_1) = f(v_1) + f(v_3).
\end{equation}
\begin{equation}
f(v_4) + f(v_6) = f(v_6) + f(v_8), \ldots, f(v_n) + f(v_2) = f(v_2) + f(v_4).
\end{equation}

From these equations, we get the following relations:
\begin{equation}\label{equ:1}
f(v_1) = f(v_5) = f(v_9) = \ldots = f(v_{n-3})
\end{equation}
\begin{equation}\label{equ:2}
f(v_2) = f(v_6) = f(v_{10}) = \ldots = f(v_{n-2})
\end{equation}
\begin{equation}\label{equ:3}
f(v_3) = f(v_7) = f(v_{11}) = \ldots = f(v_{n-1})
\end{equation}
\begin{equation}\label{equ:4}
f(v_4) = f(v_8) = f(v_{12}) = \ldots = f(v_n)
\end{equation}
Notice that Equations~\ref{equ:1},~\ref{equ:2},~\ref{equ:3} and~\ref{equ:4} contain all the vertices $v_i$, $i \in [n]$, such that $i \bmod 4 = 1$, $i \bmod 4 = 2$ , $i \bmod 4 =3$ and $i \bmod 4 = 0$, respectively. Now, we consider the following cases: 
\begin{itemize}
	\item If $n \bmod 4 = 1$, then $(n-3) \bmod 4 = 2, (n-2) \bmod 4 = 3, (n-1) \bmod 4 = 0$ and $n \bmod 4 = 1$. By the observation that Equation~\ref{equ:2} contains all the vertices $v_i$, $i \in [n]$, such that $i \bmod 4 = 2$, $f(v_{n-3}) = f(v_2)$. Similarly, $f(v_{n-2}) = f(v_3), f(v_{n-1}) = f(v_4)$ and $f(v_n) = f(v_1)$. So, by Equations~\ref{equ:1},~\ref{equ:2},~\ref{equ:3} and~\ref{equ:4}, we get that all the vertices have the same label. Thus, $S$ contains only one label, $k/2$, with $\alpha_S(k/2) = n$. 
	\item If $n \bmod 4 = 2$, then $(n-3) \bmod 4 = 3, (n-2) \bmod 4 = 0, (n-1) \bmod 4 = 1$ and $n \bmod 4 = 2$. By the observation that Equation~\ref{equ:2} contains all the vertices $v_i$, $i \in [n]$, such that $i \bmod 4 = 2$, $f(v_n) = f(v_2)$. Similarly, $f(v_{n-3}) = f(v_3)$. So, by Equations~\ref{equ:1},~\ref{equ:2},~\ref{equ:3} and~\ref{equ:4}, we get that $f(v_1) = f(v_3)$ and $f(v_2) = f(v_4)$. Denote $f(v_1) = a$ and $f(v_2) = b$. As $f(v_1) + f(v_3) = k$ and $f(v_2) + f(v_4) = k$, we get that $a = b = k/2$. Thus, $S$ contains only one label, $k/2$, with $\alpha_S(k/2) = n$.
	\item If $n \bmod 4 = 3$, then $(n-3) \bmod 4 = 0, (n-2) \bmod 4 = 1, (n-1) \bmod 4 = 2$ and $n \bmod 4 = 3$. By the observation that Equation~\ref{equ:2} contains all the vertices $v_i$, $i \in [n]$, such that $i \bmod 4 = 2$, $f(v_{n-3}) = f(v_4)$. Similarly, $f(v_{n-2}) = f(v_1), f(v_{n-1}) = f(v_2)$ and $f(v_n) = f(v_3)$. So, by Equations~\ref{equ:1},~\ref{equ:2},~\ref{equ:3} and~\ref{equ:4}, we get that all the vertices have the same label. Thus, $S$ contains only one label, $k/2$, with $\alpha_S(k/2) = n$.
	\item If $n \bmod 4 = 0$, then $(n-3) \bmod 4 = 1, (n-2) \bmod 4 = 2, (n-1) \bmod 4 = 3$ and $n \bmod 4 = 0$, so we do not get any new relation. Denote $f(v_1) = a$ and $f(v_2) = b$. As $f(v_1) + f(v_3) = k$ and $f(v_2) + f(v_4) = k$, we get $f(v_3) = k-a$ and $f(v_4) = k-b$. By Equations~\ref{equ:1},~\ref{equ:2},~\ref{equ:3} and~\ref{equ:4}, $S$ contains $4$ labels $a, b, k-a, k-b$ with $\alpha_S(a) = \alpha_S(b) = \alpha_S(k-a) = \alpha_S(k-b) = n/4$.
\end{itemize}
Conversely, first assume that $n \bmod 4 = 0$ and $S$ contains $4$ labels $a, b, k-a, k-b$ with $\alpha_S(a) = \alpha_S(b) = \alpha_S(k-a) = \alpha_S(k-b) = n/4$, for any $a,b \in \mathbb{N}$ such that $a,b < k$. Denote $V(G) = \{v_1, v_2, \ldots, v_n\}$ in the cyclic order. Let $f: V(G) \rightarrow S$ be a bijective function defined as follows. For every $i \in [n/4],$ let $ f(v_{4i - 3}) = a, f(v_{4i - 2}) = b, f(v_{4i - 1}) = k-a$ and $f(v_{4i}) = k-b$. It is easy to see that $f \in \calM$, so $G$ is $S$-fair. 

Finally, assume that $n \bmod 4 \neq 0$ and $S$ contains only one label, $k/2$, with $\alpha_S(k/2) = n$. Let $f: V(G) \rightarrow S$ be a bijective function defined as follows. For every $v \in V(G), f(v) = k/2$. It is easy to see that $f \in \calM$, so $G$ is $S$-fair.
\end{proof}

We first prove that the \fair\ problem is \FPT parameterized by $\fvs + \alpha + \Delta$. The following two lemmas will be helpful in proving it.


\begin{lemma}\label{lem:extendToForest}
There exists an $\OO(|V(G)|)$-time algorithm that, given 
\begin{inparaenum}[(i)]
	\item a graph $G$,
	\item a multiset of positive integers $S$, 
	\item an induced subgraph $F$ of $G$ such that its a forest, and
	\item a bijection $f'$ from $N_G(F) \cup \leaf(F)$ to a subset $S'$ of $S$,
	\end{inparaenum} 
returns another bijection $f''$ from $N_G(F) \cup V(F)$ to a set $S''$\footnote{$S''$ may not be a subset of $S$.} such that $S' \subseteq S''$ and $f' = f''|_{V'}$. Moreover, if $G$ is $S$-fair and $f' = f|_{N_G(F) \cup \leaf(F)}$ for some $f \in \calM$, then $f'' = f|_{N_G(F) \cup V(F)}$.
%
\end{lemma}

\begin{proof}
Let $(G,k)$ be an instance of \fair. Let $k$ be the required $S$-fairness constant. Let ${\cal F} =\{T_1, T_2, \ldots, T_t\}$ be the set of connected components of $F$. For every tree $T \in {\cal F}$, we do the following. Let $r$ be an arbitrarily chosen non-leaf vertex of $T$. Then, consider $T$ as a rooted tree with $r$ as the root vertex. Let $d$ be the depth of the tree $T$. We partition the vertex set $V(T) = V_1 \cup V_2 \ldots \cup V_d$, such that $V_i$ contains all the vertices of $T$ at depth $i$. Note that $V_1 = \{r\}$ and $V_d \subseteq \leaf(T)$. Now, consider a vertex $v \neq r$ in $T$. We partition $N_G(v) = \{p_v\} \cup \big(\child_T(v) \cap \leaf(T)\big) \cup \big(\child_T(v) \setminus \leaf(T)\big) \cup \big(N_G(v) \setminus V(T)\big)$, where $p_v$ is the parent of $v$ in $T$, $\child_T(v) \cap \leaf(T)$ is the set of children of $v$ in $T$ which are leaves of $T$, $\child_T(v) \setminus \leaf(T)$ is the set of children of $v$ in $T$ which are non-leaf vertices of $T$ and $N_G(v) \setminus T$ is the set of neighbors of $v$ not in $T$.

Given $f'$, we define another function $f_T$ on $V' = V(T) \setminus \leaf(T)$ recursively as follows.
\begin{itemize}[(i)]
	\item {\bf Base Case:} For all $v \in V'$ such that $v$ has a leaf child, let $w$ be an arbitrarily chosen leaf child of $v$. Then, $f_T(v) = k - \sum f'(N_G(w) \setminus V(T))$. 
	\item {\bf Recursive Step:} For all $v \in V'$ such that $v$ doesn't have any leaf child, let $w$ be an arbitrarily chosen child of $v$. Then, $f_T(v) = k - \sum f_T(\child_T(w) \setminus \leaf(T)) - \sum f'(\child_T(w) \cap \leaf(T)) - \sum f'(N_G(w) \setminus T) $.
\end{itemize}
Note that, if $v \in V_i$, then $\child_T(v) \in V_{i+1}$. So, we compute $f_T$ by processing vertices of $T$ in the order $V_{d-1}, \ldots, V_1$. We now define $f''$ from $N_G(F) \cup V(F)$ to $S'' = S' \cup f_{T_1} \cup f_{T_2} \ldots f_{T_t}$ as follows.
\begin{itemize}[(i)]
	\item For every $v \in N_G(F) \cup L, f''(v) = f'(v)$.
	\item For every $i \in [t]$ and $v \in V(T_i) \setminus \leaf(T_i), f''(v) = f_{T_i}(v)$.  
\end{itemize}
Clearly, $f' = f''|_{V'}$ and therefore $S' \subseteq S''$. The above recursive procedure visits every vertex of $G$ at most once, so it runs in time $\OO(|V(G)|)$.

Now, suppose that $G$ is an $S$-fair graph and $f' = f|_{N_G(F) \cup \leaf(F)}$ for some $f \in \calM$. As $f \in \calM$, for every $T \in {\cal F}$ and $v \in V(T), \sum f(N(v)) = k$. Consider a tree $T \in {\cal F}$. Let $v$ be a non-leaf vertex of $T$. Then, for all $w \in \child_T(v), \sum f(N_G(w)) = k \Rightarrow f(v) + \sum f(\child_T(w) \setminus \leaf(T)) - \sum f(\child_T(w) \cap \leaf(T)) - \sum f(N_G(w) \setminus V(T)) = k$. As $f' = f|_{N_G(F) \cup \leaf(F)}$, for all $v \in V(T) \setminus \leaf(T)$ and $w \in \child_T(v), f(v) = k - \sum f(\child_T(w) \setminus \leaf(T)) - \sum f'(\child_T(w) \cap \leaf(T)) - \sum f'(N_G(w) \setminus V(T))$. If $w$ is a leaf node, then $\child_T(w) = \emptyset$. So, we can write $f|_{N_G(F) \cup V(F)}$ as follows.

\begin{itemize}[(i)]
	\item For every $v \in N_G(F) \cup L, f(v) = f'(v)$.
	\item For every $i \in [t]$ and $v \in V(T_i) \setminus \leaf(T_i)$,
		\begin{itemize}
			\item if $v$ has a leaf child $w$, then $f(v) = k - \sum f'(N_G(w) \setminus V(T))$.
			\item else, let $w$ be a child of $v$, then $f(v) = k - \sum f(\child_T(w) \setminus \leaf(T)) - \sum f'(\child_T(w) \cap \leaf(T)) - \sum f'(N_G(w) \setminus V(T))$. 
		\end{itemize} 
\end{itemize}
As $f''$ and $f|_{N_G(F) \cup V(F)}$ have same base case and recursive step, $f'' = f|_{N_G(F) \cup V(F)}$. This also implies that $S'' \subseteq S$.
\end{proof}

The following corollary directly follows from the above lemma.

\begin{corollary}\label{cor:checkExtendingToForest}
Let $G$ be a graph and $S$ be multiset of positive integers. Let $F$ be an induced subgraph of $G$ that is a forest. Then, in time $\OO(\alpha(S)^{|N_G(F)| + |\leaf(F)|} \cdot |V(G)|)$, we can compute a superset of the set $\cal G$ of functions from $N_G(F) \cup V(F)$ to $S$ such that for every $f \in \calM$, there exists a $g \in {\cal G}$ such that $g = f|_{N_G(F) \cup V(F)}$.
\end{corollary}

\begin{proof}
As the number of different functions $f'$ from $N_G(F) \cup \leaf(F)$ to $S$ is $\OO(\alpha(S)^{|N_G(F)| + |\leaf(F)|})$, we can compute a superset of $\cal G$ by repeatedly applying Lemma~\ref{lem:extendToForest} for each such $f'$ and adding $f''$ to $\cal G$ if $S'' \subseteq S$.
\end{proof}

\begin{lemma}\label{lem:alphaDeltaStar}
The \fair\ problem is \FPT parameterized by $\alpha + \Delta$ for disjoint union of stars. 
\end{lemma}

\begin{proof}
The \FPT algorithm is based on ILP. We first give the algorithm and then prove its correctness. \smallskip\\
{\bf Algorithm:} Let $(G,k)$ be an instance of \fair\ for disjoint union of stars. Thus, $G = K_{1, n_1} + K_{1, n_2} + \ldots + K_{1, n_t}$ for some $t, n_1, n_2, \ldots, n_t \in \mathbb{N}$. Note that $\Delta = \Delta(G) = $ max$(n_1, n_2, \ldots, n_t)$. Denote $V(K_{1, n_i}) = \{v_i\} \cup B_i$ where $v_i$ is the highest degree vertex in $K_{1, n_i}$ and $B_i$ is the set of all other vertices in $K_{1, n_i}$, for every $i \in [t]$. Let $k$ be the required $S$-fairness constant of $G$. By Observations~\ref{obs:disjointMagic} and~\ref{obs:bipartiteMagic}, $G$ is $S$-fair if and only if there exists a bijection $f: V(G) \rightarrow S$ such that for all $i \in [t], f(v_i) = \sum f(B_i) = k$. So, $G$ is $S$-fair only if $\alpha_S(k) \geq t$. Let $S' = S \setminus \{s_1 = k, s_2 = k, \ldots, s_t = k\}$ and let $\widehat{\alpha} = \alpha(S')$. Let $\ell_1, \ell_2, \ldots, \ell_{\widehat{\alpha}}$ be the unique labels in $S'$. Let ${\cal D} = \{B_1, B_2, \ldots, B_t\}$. We partition $\cal D$ into ${\cal D}_1, {\cal D}_2, \ldots, {\cal D}_\Delta$ such that for every $i \in [\Delta]$, every set $B \in {\cal D}_i$ is of size $i$. As the number of unique labels are $\widehat{\alpha}$, for every $i \in [\Delta]$, any set in ${\cal D}_i$ can have at most $\widehat{\alpha}^i$ different label assignments. For every $i \in [\Delta],$ let ${\cal L}_i$ be the set of feasible label assignments for ${\cal D}_i$, i.e., the label assignments for any set in ${\cal D}_i$ for which the sum of the labels is $k$. For every $i \in [\Delta]$, every label assignment $la \in {\cal L}_i$ is a set $\{t^1_{la}, t^2_{la}, \ldots, t^{\widehat{\alpha}}_{la}\}$, where, for every $j \in [\widehat{\alpha}]$, $t^j_{la}$ denoted the number of times label $\ell_j$ is used in the label assignment $la$. For every $i \in [\Delta]$ and $la \in {\cal L}_i$, we have a variable $n_{i,la}$. For any $i \in [\Delta]$, $la \in {\cal L}_i$ and a function $f \in \calM$, $n_{i,la}$ represents the number of times label assignment $la$ is used in ${\cal D}_i$ for $f$. 
Then, the algorithm works as follows.
\begin{itemize}
	\item If $\alpha_S(k) < t$, then return {\sc False}.
	\item Solve the following ILP to find $n_{i,la}$, for every $i \in [\Delta]$ and $la \in {\cal L}_i$.
		\begin{equation}\label{equ:numberOfStars}
			\forall i \in [\Delta], \sum_{la \in {\cal L}_i} n_{i, la} = |{\cal D}_i|.
		\end{equation}
		\begin{equation}\label{equ:numberOfLabels}
			\forall j \in [\widehat{\alpha}], \sum_{i \in [\Delta]} \sum_{la \in {\cal L}_i} n_{i, la} \cdot t^j_{la} = \alpha_{S'}(\ell_j). 
		\end{equation}
		\begin{equation}
			\forall i \in [\Delta], \forall la \in {\cal L}_i;  n_{i, la} \geq 0.
		\end{equation}
	\item If the ILP returns a feasible solution, then return {\sc True}; otherwise, return {\sc False}.
\end{itemize}
\smallskip
{\bf Correctness:} Equation~\ref{equ:numberOfStars} ensures that for every $i \in [\Delta]$, the number of label assignments used in ${\cal D}_i$ is equal to the number of sets ${\cal D}_i$ has. Equation~\ref{equ:numberOfLabels} ensures that for every $j \in [\widehat{\alpha}]$, the number of times label $\ell_j$ is used is equal to the number of times it appears in $S'$. For every $i \in [\Delta]$ and for every $B \in {\cal D}_i$, all the vertices in $B$ have the same neighborhood which is just a single vertex. Thus by Observation~\ref{obs:sameNeighborhood}, it is sufficient to know the label assignment for $B$; we can arbitrarily assign labels to the vertices in $B$ once we have decided which labels to use for $B$. Keeping this interpretation in mind, we now prove the correctness, i.e. the algorithm returns {\sc True} if and only if $G$ is $S$-fair. 

Assume first that the algorithm returns {\sc True}. It means the ILP assigned non-negative integer values for the variables $n_{i,la}$, $i \in [\Delta]$ and $la \in {\cal L}_i$ such that Equations~\ref{equ:numberOfStars} and~\ref{equ:numberOfLabels} are satisfied. As for every $i \in [\Delta]$, ${\cal L}_i$ is the set of feasible label assignments, this implies that we got a label assignment for every $B \in {\cal D}$ such that sum of the labels of the vertices in $B$ is $k$. Thus, $G$ is $S$-fair.

Conversely, let $G$ be $S$-fair. Then, by Observations~\ref{obs:disjointMagic} and~\ref{obs:bipartiteMagic}, there exists a bijection $f: V(G) \rightarrow S$ such that for all $i \in [t], f(v_i) = \sum f(B_i) = k$. So, $\alpha_S(k) \geq t$. Moreover, every $B \in {\cal D}$ has a feasible label assignment so the ILP admits a feasible solution. Thus, the algorithm will return {\sc True}. As the number of variables $n_{i,la}$ is $\OO(\Delta \cdot \widehat{\alpha}^\Delta)$, by Theorem~\ref{the:runningTimeILP}, the \fair\ problem is \FPT parameterized by $\alpha + \Delta$ for disjoint union of stars. 
\end{proof}

\begin{theorem}\label{the:fvsAlphaDelta}
The \fair\ problem is \FPT parameterized by $\fvs + \alpha + \Delta$.
\end{theorem}

\begin{proof}
Let $(G,k)$ be an instance of \fair. Let $k$ be the required $S$-fairness constant. Then, we compute a minimum feedback vertex set $FVS$ of $G$ in time $\OO(5^{|FVS|}\cdot|FVS|\cdot|V(G)|^2)$, using the algorithms given by Chen {\em et~al.}~\cite{DBLP:conf/wads/ChenFLLV07}. Let $\fvs = |FVS|$ and $F = V(G) \setminus FVS$. Then by definition of feedback vertex set, $G[F]$ is a forest. Let $\cal F$ be the set of connected components of $G[F]$. So, $\cal F$ is a collection of trees. 

Let $T$ be a tree in $\cal F$. Let $v$ be a leaf vertex of $T$ such that $v$ is not connected to any vertex in $FVS$, then $\degr_G(v) = 1$. If $\degr_G(v) = 1$, then by Lemma~\ref{lem:degOneMagic}, either $G$ is a star and $G = T$ or $G$ is a disjoint union of $T$ and $G[V(G) \setminus V(T)]$. By this argument, for any tree $T \in {\cal F}$, either all the leaves of $T$ have at least one neighbor in $FVS$ or none of the leaves are connected to any vertex in $FVS$. So, we can partition $G = G_1 + G_2$, where $G_1$ is a connected graph where all the leaves of the forest $G_1[F]$ have at least one neighbor in $FVS$ and $G_2$ is a disjoint union of stars. Note that, $G_2$ is an induced subgraph of $G[F]$. By Observation~\ref{obs:disjointMagic}, $G$ is $S$-fair if any only if $G_1$ is $S_1$-fair and $G_2$ is $S \setminus S_1$-fair, for some $S_1 \subseteq S$. 

Let $L$ be the set of leaves of $G_1[F]$. As $\Delta(G_1) \leq \Delta, |L| \leq \Delta \cdot \fvs$. By Corollary~\ref{cor:checkExtendingToForest}, we can compute in time $\OO(\alpha^{(\Delta+1)\fvs}\cdot |V(G)|)$, a superset ${\cal H}$ of the set ${\cal G}$ of functions from $V(G_1)$ to $S$ such that for every $f \in \calM$, there exists a $g \in {\cal G}$ such that $g = f|_{V(G_1)}$. We can compute ${\cal G}$ from ${\cal H}$ by going over every set $h \in {\cal H}$ and checking whether $G_1$ is fair under $h$. If ${\cal G} \neq \emptyset$, then $G_1$ is $g(V(G_1))$-fair for every function $g \in {\cal G}$. Then, for every function $g \in {\cal G}$, check whether $G_2$ is $\big( S \setminus g(V(G_1)) \big)$-fair, using Lemma~\ref{lem:alphaDeltaStar}. If for some function $g \in {\cal G}$, $G_2$ is $\big( S \setminus g(V(G_1)) \big)$-fair, then by Observation~\ref{obs:disjointMagic}, $G$ is $S$-fair. Also, by Corollary~\ref{cor:checkExtendingToForest} and Lemma~\ref{lem:alphaDeltaStar}, the \fair\ problem is \FPT parameterized by $\fvs + \alpha + \Delta$.
\end{proof}

We now prove that the \fair\ problem is \FPT parameterized by $\vc + \alpha$.

\begin{theorem}\label{the:vcAlpha}
The \fair\ problem is \FPT parameterized by $\vc + \alpha$.
\end{theorem}

\begin{proof}
The \FPT algorithm is based on ILP. We first give the algorithm and then prove its correctness. \smallskip\\
{\bf Algorithm:} Let $(G,k)$ be an instance of \fair. Let $k$ be the required $S$-fair sum. Let $S'$ be the set of unique labels in $S$. Note that $|S'| = \alpha(S)$. Let $VC \subseteq V(G)$ be a vertex cover of $G$ of size $\vc$. Let $I = V(G) \setminus VC$. By the definition of vertex cover, $I$ is an independent set of $G$, i.e., no two vertices in $I$ have an edge between them. We partition $I$ into $I_1, I_2, \ldots, I_m$ such that for every $i \in [m]$, $I_i$ is an inclusion-wise maximal set of vertices in $I$ which have the same neighborhood in $G$. As $I$ is a independent set, $m \leq 2^\vc$. For every $v \in VC$, we define a binary indicator set $\{t^v_1, t^v_2, \ldots, t^v_m\}$, where $t^v_j = 1$ if $v$ is adjacent to the vertices in $I_j$, otherwise $t^v_j = 0$, for every $j \in [m]$. 

By Observation~\ref{obs:sameNeighborhood}, if $G$ is $S$-fair and if we know the labels of $VC$ under some function $f \in \calM$, then for every $i \in [m]$, it is sufficient to know the number of times every label is used in $I_i$ under $f$ to get a bijective function $f' : V(G) \rightarrow S$ such that $f' \in \calM$. Keeping this insight in mind, let $n_{i,\ell}$ be a variable whose value (to be computed below) will be interpreted as the number of times label $\ell$ is used in $I_i$ for some function $f:V(G) \rightarrow S$, for every $\ell \in S', i \in [m]$. Then, the algorithm works as follows. 
\begin{itemize}
	\item Construct the set {\cal G} containing all possible functions $g:VC \rightarrow S$. Note that $|\cal G| \leq \alpha^\vc$.
	\item For every $g \in {\cal G}$ and $\ell \in S'$, let $\alpha_g(\ell)$ denote the number of times $\ell$ appears in $g(VC)$.
	\item For every $g \in {\cal G}$:
		\begin{itemize}
			\item Solve the following ILP to find an assignment to the variables $n_{i,\ell}$, for every $i \in [m], \ell \in S'$. \\
				\begin{equation}\label{equ:VCSum}
					\forall v \in VC, \sum_{i \in [m]} t^v_i \sum_{\ell \in S'} n_{i,\ell} \cdot \ell = k 	
				\end{equation}
				\begin{equation}\label{equ:SetCardinality}
				\forall i \in [m], \sum_{\ell \in S'}n_{i,\ell} = |I_i|
				\end{equation}
				\begin{equation}\label{equ:LabelCardinality}
				\forall \ell \in S', \sum_{i \in [m]}n_{i,\ell} = \alpha_S(\ell) - \alpha_g(\ell)
				\end{equation}
				\begin{equation}
					\forall i \in [m], \forall \ell \in S';  n_{i, \ell} \geq 0.
				\end{equation}
			\item If the ILP returns a feasible solution, then return {\sc True} if the following statement holds.
				\begin{equation}\label{equ:ISSum}
					\forall i \in [m], \sum_{v \in VC} t^v_i \cdot f(v) = k
				\end{equation}
		\end{itemize}
	\item  Return {\sc False}.

\end{itemize}
\smallskip
{\bf Correctness:} Equation~\ref{equ:VCSum} ensures that for every vertex in $VC$, the neighborhood sum is $k$. Equation~\ref{equ:SetCardinality} ensures that for every $i \in [m]$, the total number of labels used in $I_i$ is equal to the size of $I_i$. Equation~\ref{equ:LabelCardinality} ensures that for every unique label $\ell \in S'$, the total number of times it is used is equal to the number of times it appear in $S$. Finally, Equation~\ref{equ:ISSum} ensures that for every vertex in $IS$, the neighborhood sum is $k$. As for every $i \in [m]$, all the vertices in $I_i$ have the same neighborhood, we only check the sum once per $I_i$. Keeping this interpretation in mind, we now prove the correctness, i.e. the algorithm returns {\sc True} if and only if $G$ is $S$-fair. 

Assume first that the algorithm returns {\sc True}. It means the ILP assigned non-negative integer values to $n_{i,\ell}$, for every $i \in [m]$ and $\ell \in S'$, such that Equations~\ref{equ:VCSum},~\ref{equ:SetCardinality} and~\ref{equ:LabelCardinality} are satisfied as well as that Equation~\ref{equ:ISSum} returned {\sc True}. As explained above, that implies that for every vertex in $G$, the neighborhood sum is the same and equals to $k$. Thus, $G$ is $S$-fair. 

Conversely, let $G$ be $S$-fair. Then, there exists a bijection $f:V(G) \rightarrow S$ such that for every $v \in V(G), \sum f(N(v)) = k$. As $\cal G$ is the set of all possible functions from $VC$ to $S$, $f|_{VC} \in {\cal G}$, and hence there exists an iteration where the algorithm examines $g = f|_{VC}$. For $g = f|_{VC}$, the ILP admits a feasible solution. Moreover, Equation~\ref{equ:ISSum} then holds. Thus, the algorithm will return {\sc True}. As $|\cal G| \leq \alpha^\vc$ and the number of variables $n_{i,\ell}$ is at most $2^\vc \cdot \alpha$, by Theorem~\ref{the:runningTimeILP}, the \fair\ problem is \FPT parameterized by $\vc + \alpha$.
\end{proof} 

We now prove that the \fair\ problem is \FPT parameterized by $\fvs$ for regular graphs. 

\begin{theorem}\label{the:fvsrRegular}
The \fair\ problem is \FPT parameterized by $\fvs$ for regular graphs.
\end{theorem}

\begin{proof}
Let $r \in \mathbb{N}$. Let $(G,k)$ be an instance of \fair\ for $r$-regular graphs. Let $k$ be the required $S$-fairness constant. Then, we compute a minimum feedback vertex set $FVS$ of $G$ in time $\OO(5^{|FVS|}\cdot|FVS|\cdot|V(G)|^2)$, using the algorithms given by Chen {\em et~al.}~\cite{DBLP:conf/wads/ChenFLLV07}. Let $\fvs = |FVS|$ and $F = V(G) \setminus FVS$. Then, $G[F]$ is a forest. We distinguish $G$ into the following three cases based on the value of $r$.\smallskip\\
{\bf Case $1$ [$r = 1$]:} In this case, $G$ is a collection of edges, i.e., $G = tP_2$ for some $t \in \mathbb{N}$. Then, by Observation~\ref{obs:disjointMagic} and by the definition of $S$-fair labeling, $G$ is $S$-fair if and only if $S = \biguplus_{i \in [t]}\{a_i, k-a_i\}$ where for all $i \in [t], a_i \in \{1, 2, \ldots, k\}$. So, we can solve \fair\ problem in $\OO(|S|\log S) = \OO(|V(G)|\log |V(G)|)$ time by sorting $S$ and checking whether $S$ satisfies the above property. \smallskip\\
{\bf Case $2$ [$r = 2$]:} In this case, $G$ is a collection of cycles, i.e. $G = C_{n_1} + C_{n_2} + \ldots + C_{n_t}$ for some $t \in \mathbb{N}$. By the definition of feedback vertex set and the minimality of $FVS$, every cycle contains exactly one vertex from $FVS$. So, $t = \fvs$. Also, by Observation~\ref{obs:disjointMagic} and by Lemma~\ref{lem:cycleDistanceMagic}, for every $f \in \calM$, every cycle is assigned at most $4$ distinct labels, by $f$. So, $\alpha(S) \leq 4\fvs$. As $\Delta(G) = 2$ and $\alpha \leq 4\fvs$, by Theorem~\ref{the:fvsAlphaDelta}, the \fair\ problem is \FPT parameterized by $\fvs$ in this case.\smallskip\\
{\bf Case $3$ [$r \geq 3$]:} As $G[F]$ is a forest, $|E(F)| \leq |V(F)| - 1 \leq |V(G)| - 1$. Also, $G$ is a $r$-regular graph so  $|E(G)| = r \cdot |V(G)|/2$. This implies that, at least $(r/2-1)\cdot |V(G)| +1$ edges are incident to vertices of $FVS$. As every vertex of $FVS$ is incident to $r$ edges, $(r/2-1)\cdot |V(G)| +1 \leq r \cdot \fvs$. Since $r \geq 3$, we get that $|V(G)| = \OO(fvs)$. As the \fair\ problem can always be solved in time $\OO(|V(G)|!)$ using brute-force approach by going over all the permutations of labels, the \fair\ problem is \FPT parameterized by $\fvs$ in this case as well.
\end{proof}

Finally, we give a simple lemma proving that the \fair\ problem is \FPT parameterized by $\vc + \Delta$.

\begin{lemma}
The \fair\ problem is \FPT parameterized by $\vc + \Delta$.
\end{lemma}

\begin{proof}
Let $G$ be a graph of maximum degree $\Delta$ and let $VC \subseteq V(G)$ be a vertex cover of $G$ of size $\vc$. Then, by the definition of vertex cover, it is easy to see that $|V(G)| \leq \vc \cdot \Delta$. As the \fair\ problem can always be solved in time $\OO(|V(G)|!)$ using brute-force approach by going over all the permutations of labels, the \fair\ problem is \FPT parameterized by $\vc + \Delta$.
\end{proof}
\section{Conclusion and Future Research}\label{sec:conclusion}

In this paper, we initiated a systematic algorithmic study of \fair\ and presented a comprehensive picture of the parameterized complexity of the problem. We showed \textsf{NP}-hardness results on special graph classes, which implied that the problem is \paraH with respect to several combinations of structural graph parameters. We also showed that the problem is \FPT for some combinations of structural graph parameters. 

While our work is comprehensive, we stress that it also opens a whole new world of research questions within computational social choice. For illustration, let us mention a few such questions: 
\begin{enumerate}[(1)]
	\item Establishing the parameterized complexity of \fair\ with respect to $\tw + \Delta$.
	\item Establishing the parameterized complexity of \fair\ with respect to $\tw + \Delta + \alpha$. By employing the standard dynamic programming technique over tree decomposition, we can show that the problem is \FPT with respect to $\tw + \Delta$, when $\alpha$ is a constant. But, the parameterized complexity when $\alpha$ is not a constant is still open.
	\item Establishing the classical complexity of \fair\ where $S = [n]$, ,where $n$ denotes the number of vertices in the input graph.
	\item Studying the scenario where there is no input infrastructure graph and the objective is to construct one which is $S$-fair.
	\item Analysis of related labellings such as \{0,1\}-\fair~\cite{beena2009and} and vertex-bimagic labeling~\cite{babujee20141}.
	\item Introducing additional fairness notions for non-eliminating tournaments, perhaps by refining/extending/modifying the notion of $S$-fairness.
	\item Introducing manipulation and bribery to \fair.
\end{enumerate}
\newpage

\bibliographystyle{siam}
\bibliography{Refs}

\end{document}